\documentclass[a4paper]{llncs}

\newcommand{\set}[1]{\{ #1 \}}
\newcommand{\val}{\mathit{val}}

\newcommand{\N}{\mathbb{N}}

\newcommand{\G}{\mathrm{G}}
\newcommand{\M}{\mathcal{M}}

\renewcommand{\S}{\mathbb{S}}

\newcommand{\Until}{\ \mathrm{Until}\ }

\newcommand{\rank}{\textrm{rank}}

\newcommand{\summ}{\textrm{sum}}

\newcommand{\up}{\mu}

\newcommand{\co}{\textrm{Co}}

\newcommand{\buc}{\textrm{B\"uchi}}
\newcommand{\safe}{\textrm{Safe}}
\newcommand{\reach}{\textrm{Reach}}
\newcommand{\parity}{\textrm{Parity}}
\newcommand{\mem}{\textrm{mem}}

\newcommand{\VE}{V_{E}}
\newcommand{\VA}{V_{A}}
\newcommand{\W}{\mathcal{W}}
\newcommand{\WE}{\W_{E}}

\newtheorem{fact}[theorem]{Fact}

\usepackage{amssymb}
\usepackage{amsmath}
\usepackage{amsfonts}
\usepackage{graphicx}
\usepackage{framed}

\title{Trading Bounds for Memory in\\
Games with Counters}

\author{Nathana\"el Fijalkow\inst{1,2} \and Florian Horn\inst{1} \and Denis Kuperberg\inst{2,3} \and Micha{\l} Skrzypczak\inst{1,2}}
\institute{LIAFA, Universit{\'e} Paris 7 \and Institute of Informatics, University of Warsaw \and Onera/DTIM, Toulouse and IRIT, University of Toulouse}

\begin{document}

\maketitle

\begin{abstract}
We study two-player games with counters, where the objective of the first player is that the counter values remain bounded.
We investigate the existence of a trade-off between the size of the memory and the bound achieved on the counters,
which has been conjectured by Colcombet and Loeding.

We show that unfortunately this conjecture does not hold: there is no trade-off between bounds and memory,
even for finite arenas.
On the positive side, we prove the existence of a trade-off for the special case of thin tree arenas.
This allows to extend the theory of regular cost functions over thin trees,
and obtain as a corollary the decidability of cost monadic second-order logic over thin trees.
\end{abstract}

\section{Introduction}
\label{sec:intro}
This paper studies finite-memory determinacy for games with counters.
The motivation for this investigation comes from the theory of regular cost functions,
which we discuss now.

\medskip
\textbf{Regular cost functions.} The theory of regular cost functions is a \textit{quantitative} extension of the notion of regular languages,
over various structures (words and trees).
More precisely, it expresses \textit{boundedness questions}. A typical example of a boundedness question is:
given a regular language $L \subseteq \set{a,b}^*$, does there exist a bound $N$ such that all words from $L$ contain at most $N$ occurences of $a$?

This line of work has already a long history: it started in the 80s, when Hashiguchi, 
and then later Leung, Simon and Kirsten solved the \textit{star-height problem} 
by reducing it to boundedness questions~\cite{Hashiguchi90,Simon94,Leung91,Kirsten05}.
Both the logics MSO+$\mathbb{U}$ and later cost MSO (as part of the theory of regular cost functions)
emerged in this context~\cite{Bojanczyk04,BojanczykColcombet06,Colcombet09,Colcombet13a,ColcombetLoeding10},
as quantitative extensions of the notion of regular languages allowing to express boundedness questions.

Consequently, developing the theory of regular cost functions comes in two flavours:
the first is using it to reduce various problems to boundedness questions,
and the second is obtaining decidability results for the boundedness problem for cost MSO
over various structures.

For the first point, many problems have been reduced to boundedness questions.
The first example is the star-height problem over words~\cite{Kirsten05} and over trees~\cite{ColcombetLoeding10},
followed for instance by the boundedness question for fixed points of monadic formulae 
over finite and infinite words and trees~\cite{BlumensathOttoWeyer14}.
The most important problem that has been reduced is to decide the Mostowski hierarchy for infinite trees~\cite{ColcombetLoeding08}.

For the second point, it has been shown that over finite words and trees, 
a significant part of the theory of regular languages can successfully be extended to the theory of regular cost functions, 
yielding notions of regular expressions, automata, semigroups and logics that all have the same expressive power,
and that extend the standard notions.
In both cases, algorithms have been constructed to answer boundedness questions.

However, extending the theory of regular cost functions to infinite trees seems to be much harder,
and the major open problem there is the decidability of cost MSO over infinite trees. 

\medskip
\textbf{LoCo conjecture.}
Colcombet and Loeding pointed out that the only missing point to obtain 
the decidability of cost MSO is a finite-memory determinacy result for games with counters.
More precisely, they conjectured that there exists a trade-off between the size of the memory and the bound achieved on the counters~\cite{Colcombet13}.
So far, this conjecture resisted both proofs and refutations, 
and the only non-trivial positive case known is due to Vanden Boom~\cite{VandenBoom11},
which implied the decidability of the weak variant of cost MSO over infinite trees, later generalized to quasi-weak cost MSO in~\cite{BCKPV14}.
Unfortunately, weak cost MSO is strictly weaker than cost MSO, and this leaves open
the question whether cost MSO is decidable.

\medskip
\textbf{Contributions.}
In this paper, we present two contributions:
\begin{itemize}
	\item There is no trade-off, even for finite arenas, which disproves the conjecture,
	\item There is a (non-elementary) trade-off for the special case of thin tree arenas.
\end{itemize}
Our first contribution does not imply the undecidability of cost MSO, 
it rather shows that proving the decidability will involve subtle combinatorial arguments that are yet to be understood.
As a corollary of the second contribution, we obtain the decidability of cost MSO over thin trees.

\medskip
\textbf{Structure of this document.}
The definitions are given in Section~\ref{sec:defs}.
We state the conjecture in Section~\ref{sec:conjecture}.
The Section~\ref{sec:lower_bound} disproves the conjecture.
The Section~\ref{sec:thin_tree} proves that the conjecture holds for the special case of thin tree arenas.

\section{Definitions}
\label{sec:defs}
\medskip \textbf{Arenas.} The games we consider are played by two players, Eve and Adam, over potentially infinite graphs
called arenas\footnote{We refer to~\cite{LNCS2500} for an introduction to games.}.
Formally, an arena~$G$ consists of a directed graph~$(V, E)$
whose vertex set is divided into vertices controlled by Eve ($\VE$)
and vertices controlled by Adam ($\VA$).
A token is initially placed on a given initial vertex $v_0$, and the player
who controls this vertex pushes the token along an edge, reaching a new vertex;
the player who controls this new vertex takes over, and this interaction goes on
forever, describing an infinite path called a \textit{play}.
Finite or infinite plays are paths in the graphs, seen as sequences of edges, typically denoted~$\pi$.
In its most general form, a strategy for Eve is a mapping $\sigma : E^* \cdot \VE \to E$,
which given the history played so far and the current vertex picks the next edge.
We say that a play $\pi = e_0 e_1 e_2 \ldots$ is consistent with $\sigma$ if $e_{n+1} = \sigma(e_0 \cdots e_n \cdot v_n)$ 
for every~$n$ with $v_n \in \VE$. 

\medskip \textbf{Winning conditions.} A winning condition for an arena is a set of a plays for Eve, which are called the winning plays for Eve 
(the other plays are winning for Adam). 
A strategy for Eve is winning for a condition, or ensures this condition, if all plays consistent with the strategy belong to the condition. 
For a winning condition $W$, we denote $\WE(W)$ the winning region of Eve, \textit{i.e.} the set of vertices from which Eve has a winning strategy.

Here we will consider the classical parity condition as well as \emph{quantitative} bounding conditions.

The parity condition is specified by a colouring function $\Omega : V \to \set{0,\ldots,d}$,
requiring that the maximum color seen infinitely often is even.
The special case where $\Omega : V \to \set{1,2}$ corresponds to B\"uchi conditions,
denoted $\buc(F)$ where $F = \set{v \in V \mid \Omega(v) = 2}$.
We will also consider the simpler conditions $\safe(F)$ and $\reach(F)$, for $F \subseteq V$:
the first requires to avoid $F$ forever, and the second to visit a vertex from $F$ at least once.

The bounding condition $B$ is actualy a family of winning conditions with an integer parameter $B = \{B(N)\}_{N \in \N}$.
We call it a quantitative condition because it is monotone: if $N < N'$, all the plays in $B(N)$ also belong to $B(N')$.

The counter actions are specified by a function $c : E \to \set{\varepsilon,i,r}^k$,
where $k$ is the number of counters: 
each counter can be incremented ($i$), reset ($r$), or left unchanged ($\varepsilon$). 
The value of a play $\pi$, denoted $\val(\pi)$, is the supremum of the value of all counters along the play. 
It can be infinite if one counter is unbounded. 
The condition $B(N)$ is defined as the set of play whose value is less than $N$.

In this paper, we study the condition $B$-parity, where the winning condition is the intersection of a bounding condition and a parity condition.
The value of a play that satisfies the parity condition is its value according to the bounding condition. 
The value of a play which does not respect the parity condition is $\infty$.
We often consider the special case of $B$-reachability conditions, denoted $B \Until F$. In such cases, we assume that the game stops when it reaches $F$.

Given an initial vertex~$v_0$, the value $\val(v_0)$ is:
$$\inf_{\sigma}\ \sup_{\pi}\ \set{\val(\pi) \mid \pi \textrm{ consistent with } \sigma \textrm{ starting from } v_0}\ .$$

\medskip \textbf{Finite-memory strategies.} A \emph{memory structure} $\M$ for the arena $\G$ consists of a set $M$ of memory states, 
an initial memory state $m_0 \in M$ and an update function $\up: M \times E \to M$.
The update function takes as input the current memory state and the chosen edge to compute the next memory state,
in a deterministic way.
It can be extended to a function $\up: E^* \cdot V \to M$ by defining $\up^*(v) = m_0$ and $\up^* (\pi \cdot (v,v')) = \up(\up^*(\pi \cdot v), (v,v'))$.

Given a memory structure $\M$, a strategy is induced by a next-move function $\sigma: \VE \times M \to E$, 
by $\sigma(\pi \cdot v) = \sigma(v, \up^*(\pi \cdot v))$.
Note that we denote both the next-move function and the induced strategy $\sigma$.
A strategy with memory structure $\M$ has finite memory if $M$ is a finite set.
It is \emph{memoryless}, or \emph{positional} if $M$ is a singleton: it only depends on the current vertex. 
Hence a memoryless strategy can be described as a function $\sigma: \VE \to E$.

An arena $\G$ and a memory structure $\M$ for $\G$ induce the expanded arena 
$\G \times \M$ where the current memory state is stored explicitly along the current vertex:
the vertex set is $V \times M$, the edge set is $E \times \up$, defined by:
$((v,m), (v',m')) \in E'$ if $(v,v') \in E$ and $\up(m,(v,v')) = m'$.
There is a natural one-to-one correspondence between memoryless strategies in $\G \times \M$
and strategies in $\G$ using $\M$ as memory structure.

\section{The conjecture}
\label{sec:conjecture}

In this section, we state the conjecture~\cite{Colcombet13},
and explain how positive cases of this conjecture imply the decidability of cost MSO.

\subsection{Statement of the conjecture}
\begin{center}
\begin{framed}
There exists $\mem : \N^2 \to \N$ and $\alpha : \N^3 \to \N$ such that\\
for all $B$-parity games with $k$ counters, $d+1$ colors and initial vertex $v_0$,\\
there exists a strategy $\sigma$ using $\mem(d,k)$ memory states,
ensuring $B(\alpha(d,k,\val(v_0))) \cap \parity(\Omega)$.
\end{framed}
\end{center}

The function $\alpha$ is called a trade-off function:
if there exists a strategy ensuring $B(N) \cap \parity(\Omega)$,
then there exists a strategy with \textit{small} memory
that ensures $B(\alpha(d,k,N)) \cap \parity(\Omega)$.
So, at the price of increasing the bound from $N$ to $\alpha(d,k,N)$,
one can use a strategy using a small memory structure.

To get a better understanding of this conjecture, we show three simple facts:
\begin{enumerate}
	\item why reducing memory requires to increase the bound,
	\item why the memory bound $\mem$ depends on the number of counters $k$,
	\item why a weaker version of the conjecture holds, where $\mem$ depends on the value,
\end{enumerate}

\begin{figure}[!ht]
\begin{center}
\includegraphics[scale=1]{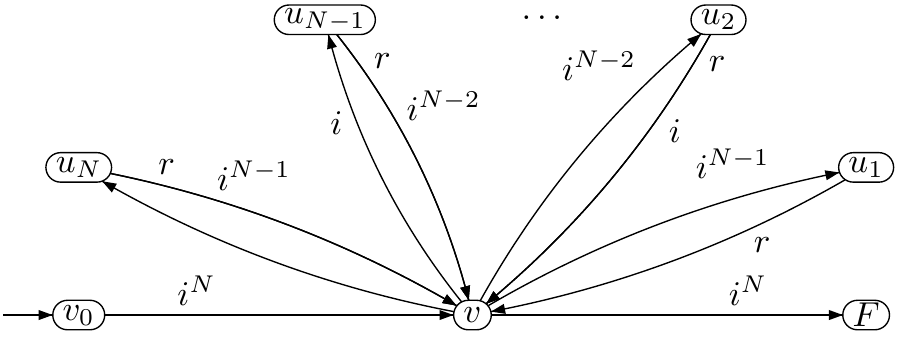}
\caption{\label{fig:trade-off_necessary} A trade-off is necessary.}
\end{center}
\end{figure}
For the first point, we present a simple game, represented in Figure~\ref{fig:trade-off_necessary}.
It involves one counter and the condition $B \Until F$.
Starting from $v_0$, the game moves to $v$ and sets the value of the counter to $N$.
The objective of Eve is to take the edge to the right to $F$.
However, this costs $N$ increments, so if she wants the counter value to remain smaller than $N$
she has to set its value to $0$ before taking this edge.
She has $N$ options: for $\ell \in \set{1,\ldots,N}$, the $\ell$\textsuperscript{th} option consists in going to $u_\ell$, involving the following actions:
\begin{itemize}
	\item first, take $N-\ell$ increments,
	\item then, reset the counter,
	\item then, take $\ell-1$ increments, setting the value to $\ell-1$.
\end{itemize}
It follows that there is a strategy for Eve to ensure $B(N) \Until F$,
which consists in going successively through $u_N$, $u_{N-1}$, and so on, until $u_1$, and finally to $F$.
Hence to ensure that the bound is always smaller than $N$, Eve needs $N+1$ memory states.

However, if we consider the bound $2N$ rather than $N$, then Eve has a very simple strategy, which consists in 
going directly to $F$, using no memory at all.
This is a simple example of a trade-off: to ensure the bound $N$, Eve needs $N+1$ memory states,
but to ensure the worse bound $2N$, she has a positional strategy.


\smallskip
For the second point, consider the following simple game with $k$ counters (numbered cyclically) 
and only one vertex, controlled by Eve. 
There are $k$ self-loops, each incrementing a counter and resetting the previous one.
Eve has a simple strategy to ensure $B(1)$, which consists in cycling through the loops,
and uses $k$ memory states.
Any strategy using less than $k$ memory states ensures no bound at all,
as one counter would be incremented infinitely many times but never reset.
It follows that the memory bound $\mem$ in the conjecture has to depend on $k$,
and this example shows that $\mem \ge k$ is necessary.

\smallskip
For the third point, we give an easy result that shows the existence of finite memory strategies
whose size depend on the value, even without losing anything on the bound.
Unfortunately, this statement is not strong enough; we discuss in the next subsection the implications of the conjecture.
\begin{lemma}
\label{lem:finite_memory_trivial}
For all $B$-parity games with $k$ counters and initial vertex $v_0$,
there exists a strategy $\sigma$ ensuring $B(\val(v_0)) \cap \parity(\Omega)$ 
with $(\val(v_0)+1)^k$ memory states.
\end{lemma}

\begin{proof}
We consider the memory structure $\M = (\set{0,\ldots,N}^k,0^k,\up)$ which keeps track of the counter values,
where $N = \val(v_0)$. 
We construct the arena $\G \times \M$ and add a new vertex $\bot$ which is reached if a counter reaches
the value $N+1$, according to the memory structure.
The condition $\safe(\bot) \cap \parity(\Omega)$,
which requires never to reach $\bot$ and to satisfy the parity condition, is equivalent to $B(N) \cap \parity(\Omega)$.
Since parity games are positionally determined, there exists a positional strategy ensuring $\safe(\bot) \cap \parity(\Omega)$,
which induces a finite-memory strategy using $\M$ as memory structure ensuring $B(N) \cap \parity(\Omega)$.
\hfill\qed\end{proof}

\subsection{The interplay with cost MSO}
The conjecture stated above has a purpose: if true,
it implies the decidability of cost MSO over infinite trees.
More precisely, the technical difficulty to develop the theory of regular cost functions over infinite trees
is to obtain effective constructions between variants of automata with counters,
and this is what this conjecture is about.

In the qualitative case (without counters), to obtain the decidability of MSO over infinite trees,
known as Rabin's theorem~\cite{Rabin69}, one transforms MSO formulae into equivalent automata.
The complementation construction is the technical cornerstone of this procedure.
The \textit{key} ingredient for this is games, and specifically positional determinacy for parity games. 
Similarly, other classical constructions, to simulate either two-way or alternating automata
by non-deterministic ones, make a crucial use of positional determinacy for parity games.

In the quantitative case now, Colcombet and Loeding~\cite{ColcombetLoeding10} showed that
to extend these constructions, one needs a similar result on parity games with counters,
which is the conjecture we stated above.

\smallskip
So far, there is only one positive instance of this conjecture, which is the special case of B-B\"uchi games
over chronological arenas\footnote{The definition of chronological arenas is given in Section~\ref{sec:thin_tree}.}.

\begin{theorem}[\cite{VandenBoom11}]
\label{thm:buchi}
For all B-B\"uchi games with $k$ counters and initial vertex $v_0$ over chronological arenas,
Eve has a strategy ensuring $B(2 \cdot \val(v_0)) \cap \buc(F)$ with $2 \cdot k!$ memory states.
\end{theorem}

It leads to the following decidability result.

\begin{corollary}[\cite{VandenBoom11}]
Weak cost MSO over infinite trees is decidable.
\end{corollary}

\section{No trade-off over Finite Arenas}
\label{sec:lower_bound}
In this section, we show that the conjecture does not hold, even for finite arenas.

\begin{theorem}
For all $K$, for all $N$, there exists a finite $B$-reachability game $G_{K,N}$ with one counter such that:
\begin{itemize}
	\item there exists a $3^K$ memory states strategy ensuring $B(K(K+3)) \Until F$,
	\item no $K+1$ memory states strategy ensure $B(N) \Until F$.
\end{itemize}
\end{theorem}

We proceed in two steps. The first is an example giving a lower bound of $3$, and the second is a nesting of this first example.

\subsection{A first lower bound of $3$}
We start with a game $\G_1$, which gives a first lower bound of $3$.
It is represented in Figure~\ref{fig:lower_bound_3mem}.
The condition is $B \Until F$.
In this game, Eve is torn between going to the right to reach $F$, which implies incrementing the counter,
and going to the left, to reset the counter.
The actions of Eve from the vertex $u_N$ are:
\begin{itemize}
	\item \textit{increment}, and go one step to the right, to $v_{N-1}$,
	\item \textit{reset}, and go two steps to the left, to $v_{N+2}$.
\end{itemize} 
The actions of Adam from the vertex $v_N$ are:
\begin{itemize}
	\item \textit{play}, and go down to $u_N$,
	\item \textit{skip}, and go to $v_{N-1}$.
\end{itemize}

\begin{figure}[!ht]
\begin{center}
\includegraphics[scale=1]{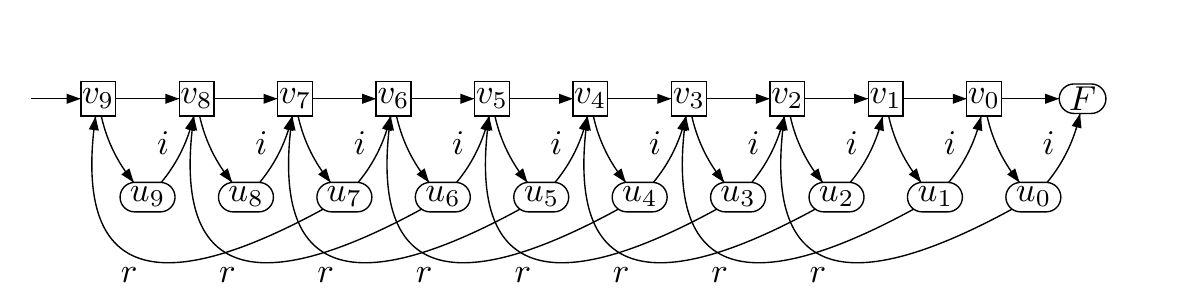}
\caption{\label{fig:lower_bound_3mem} Part of the game $\G_1$, where Eve needs $3$ memory states.}
\end{center}
\end{figure}
Formally:
$$V = \left\{
\begin{array}{l}
\VE = \set{u_n \mid n \in \N} \\
\VA = \set{v_n \mid n \in \N}
\end{array}\right.$$

$$E = \left\{
\begin{array}{llr}
	 & \set{v_{n+1} \xrightarrow{\ \ } v_n \mid n \in \N} 					\\
\cup & \set{v_n \xrightarrow{\ \ } u_n \mid n \in \N} 						\\
\cup & \set{u_{n+1} \xrightarrow{\ i\ } v_n \mid n \in \N} \cup \set{u_0 \xrightarrow{\ i\ } F}  \\
\cup & \set{u_n \xrightarrow{\ r\ } v_{n+2} \mid n \in \N} 					
\end{array}\right.$$

\begin{theorem}
\label{thm:G1}
In $\G_1$:
\begin{itemize}
	\item Eve has a $4$ memory states strategy ensuring $B(3) \Until F$, 
	\item Eve has a $3$ memory states strategy ensuring $B(4) \Until F$,
	\item For all $N$, no $2$ memory states strategy ensures $B(N) \Until F$ from $v_N$.
\end{itemize}
\end{theorem}

The first item follows from Lemma~\ref{lem:finite_memory_trivial}. However, to illustrate the properties of the game $\G_1$ we will provide a concrete strategy with $4$ memory states that ensures $B(3) \Until F$.
The memory states are $i_1,i_2,i_3$ and $r$, linearly ordered by $i_1 < i_2 < i_3 < r$.
With the memory states $i_1,i_2$ and $i_3$, the strategy chooses to increment,
and updates its memory state to the next memory state.
With the memory state $r$, the strategy chooses to reset,
and updates its memory state to $i_1$.
This strategy satisfies a simple invariant: it always resets to the right of the previous reset, if any.

\begin{figure}[!ht]
\begin{center}
\includegraphics[scale=.8]{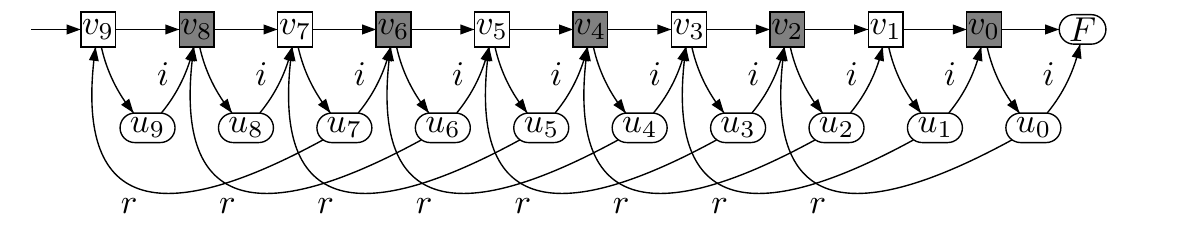}
\caption{\label{fig:3mem} Illustration of the $3$ memory states strategy in $\G_1$.}
\end{center}
\end{figure}
\medskip
We show how to save one memory state, at the price of increasing the bound by one:
we construct a $3$ memory states strategy ensuring $B(4) \Until F$.
The idea, as represented in Figure~\ref{fig:3mem}, is to color every second vertex and to use this information
to track progress.
The $3$ memory states are called $i$, $j$ and $r$. The update is as follows: the memory state is unchanged in uncoloured (white) states,
and switches from $i$ and $j$ and from $j$ to $r$ on gray states.
The strategy is as follows: in the two memory states $i$ and $j$, Eve chooses to increment,
and in $r$ she chooses to reset.
As for the previous strategy, this strategy ensures that it always resets to the right of the previous reset, if any.

\medskip
We now show that $2$ memory states is not enough. 
Assume towards contradiction that there exists a $2$ memory states strategy ensuring $B(N) \Until F$ from $v_{2N}$, for some $N$,
using the memory structure $\M = (M, \up, m_0)$.

We first argue that without loss of generality we can assume that the strategy $\sigma$ is normalized,
\textit{i.e.} satisfies the following three properties:
\begin{enumerate}
	\item for all $n \le 2N$, there is at least one memory state that chooses increment from $u_n$,
	\item for all $n \le 2N$ but at most $N$ of them, there is at least one memory state that chooses reset from $u_n$,
	\item no play from $(v_{2N},m_0)$ consistent with $\sigma$ comes back to $v_{2N}$.
\end{enumerate}
Indeed:
\begin{enumerate}
	\item Assume towards contradiction that this is not the case, then there exists $n$ such that Eve resets from $u_n$ with both memory states;
	Adam can loop around this $u_n$, contradicting that $\sigma$ ensures to reach $F$.
	\item Assume towards contradiction that there are at least $N+1$ vertices $u_n$ from which Eve increments from $u_n$ with both memory states;
	Adam can force $N+1$ increments without a reset, contradicting that $\sigma$ ensures $B(N)$.
	\item For $m$ and $m'$ two memory states, we say that $m < m'$ if there exists a play from $(v_{2N},m')$ consistent with $\sigma$ which reaches $(v_{2N},m)$.
	Since $\sigma$ ensures to reach $F$, the graph induced by $<$ is acyclic. 
	We can take as initial memory state from $v_{2N}$ the smallest memory state which is smaller or equal to $m_0$.
\end{enumerate}

We fix the strategy of Adam which skips if, and only if, both memory states of $\sigma$ choose to increment.
Consider the play from $v_{2N}$ consistent with $\sigma$ and this strategy of Adam.
This means that for all vertices $u_n$ that are reached, there is one memory state that resets, and one that increments.
Since $\sigma$ ensures $B(N)$ and there are at most $N$ positions skipped, at some point Eve chooses to reset.
From there two scenarios are possible:
\begin{itemize}
	\item Either Eve keeps resetting until she reaches $v_{2N}$, contradicting that $\sigma$ is normalized,
	\item Or she starts incrementing again, which means that she uses the same memory state than she did before the reset,
	implying that there is a loop, contradicting that $\sigma$ ensures to reach $F$.
\end{itemize}

\subsection{General lower bound}
We now push the example above further.

A first approach is to modify $\G_1$ by increasing the length of the resets, going $\ell$ steps to the left rather than only $2$.
However, this does not give a better lower bound: there exists a $3$ memory states strategy in this modified game
that ensures twice the value, following the same ideas as presented above.

\begin{figure}[!ht]
\begin{center}
\includegraphics[scale=1]{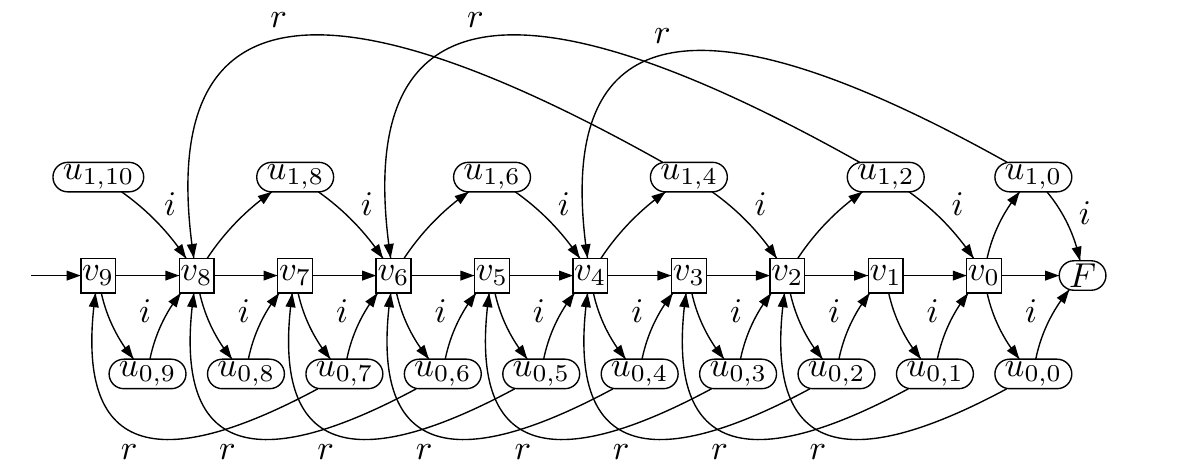}
\caption{\label{fig:G_22} The game with two levels.}
\end{center}
\end{figure}

We construct $\G_{K,N}$, a nesting of the game $\G_1$ with $K$ levels.
Unlike $\G_1$, it is finite, as we only keep a long enough ``suffix''.
In figure~\ref{fig:G_22}, we represented the interaction between two levels.
Roughly speaking, the two levels are independent, so we play both games at the same time.
Those two games use different timeline.
For instance, in Figure~\ref{fig:G_22}, the bottom level is based on $(+1,-2)$ (an increment goes one step to the right,
a reset two steps to the left), and the top level is based on $(+2,-4)$.
This difference in timeline ensures that a strategy for Eve needs to take care somehow independently of each level,
ensuring that the number of memory states depends on the number of levels.

To give the formal definition of $\G_{K,N}$, we need two functions,
$d(K,N) = (N+1)^{K-1}$ and 
$$n(K+1,N) = 
\begin{cases}
2N & \textrm{ if } K = 0, \\
(N+1)^{K+1} + (N+1) \cdot n(K,N) & \textrm{ otherwise}. \\
\end{cases}$$

We now define $\G_{K,N}$.
$$V = \left\{
\begin{array}{l}
\VE = \set{u_{p,n} \mid p \in \set{1,\ldots,K}, n \le n(K,N)} \\
\VA = \set{v_n \mid n \le n(K,N)}
\end{array}\right.$$

$$E = \left\{
\begin{array}{ll}
	 & \set{ v_{n+1} \xrightarrow{\ } v_n \mid n	} \\
\cup & \set{v_n \xrightarrow{\ } u_{p,n} \mid p,n } \\
\cup & \set{u_{p,n + d(p,N)} \xrightarrow{\ i\ } v_n \mid p,n } \\
\cup & \set{u_{p,0} \xrightarrow{\ i\ } F \mid p } \\
\cup & \set{u_{p,n} \xrightarrow{\ r\ } v_{n + (p+1) \cdot d(p,N)} \mid p, n }
\end{array}\right.$$

Observe that $\G_{1,N}$ is the ``suffix'' of length $n(1,N)$ of $\G_1$, for all $N$.

\begin{theorem}
\label{thm:GKN}
In $\G_{K,N}$:
\begin{itemize}
	\item Eve has a $3^K$ memory states strategy ensuring $B(K(K+3)) \Until F$,
	\item No $K+1$ memory states strategy ensures $B(N) \Until F$ from $v_{n(K,N)}$.
\end{itemize}
\end{theorem}

We first construct a strategy with $3^K$ memory states ensuring $B(K(K+3)) \Until F$.
To this end, we construct for the $p$\textsuperscript{th} level a strategy with $3$ memory states ensuring $B(2(p+1)) \Until F$,
using the same ideas as for $\G_1$, colouring every $(p+1) \cdot d(p,N)$ vertices.
Now we construct the general strategy by playing independently in each copy,
except that when a reset is taken, all memory structures update to the (initial) memory state $i$.
This way, it ensures that it always resets to the right of the previous reset, if any.
It uses $3^K$ memory states, and ensures $B(\sum_{p = 1}^K 2 \cdot (p+1)) \Until F$, \textit{i.e.} $B(K(K+3)) \Until F$.

\medskip
We now show that $K+1$ memory states is not enough.
We proceed by induction on $K$. 
The case $K = 1$ follows from Theorem~\ref{thm:G1}.

Consider a strategy ensuring $B(N) \Until F$ from $v_{n(K+1,N)}$ in $\G_{K+1,N}$, for some $N$,
using the memory structure $\M = (M, \up, m_0)$.
We will prove that it has at least $K+2$ memory states.
To this end, we will show that it implies a strategy ensuring $B(N) \Until F$ in $\G_{K,N}$, which uses one less memory state.
The induction hypothesis will conclude.

We first argue that without loss of generality we can assume that no play from $(v_{n(K+1,N)},m_0)$ consistent with $\sigma$ comes back to $v_{n(K+1,N)}$.
(The proof is the same as for $\G_1$.)
For $m$ and $m'$ two memory states, we say that $m < m'$ if there exists a play from $(v_{n(K+1,N)},m')$ consistent with $\sigma$ which reaches $(v_{n(K+1,N)},m)$.
Since $\sigma$ ensures to reach $F$, the graph induced by $<$ is acyclic. 
We can take as initial memory state from $v_{n(K+1,N)}$ the smallest memory state which is smaller or equal to $m_0$.

We now argue that there exists $n \le n(K+1,N)$ 
and a play from $(v_n,m)$ (for some memory state $m \in M$) to $u_{K+1,n - n(K,N)}$ consistent with $\sigma$,
which does not use the topmost level (level $K+1$), and such that from there $\sigma$ chooses to reset.

Assume towards contradiction that this is not the case.
Consider the following strategy of Adam, from $v_{n(K+1,N)}$. 
It alternates ($N+1$ times) between skipping for $n(K,N)$ steps and going to the topmost level.
By assumption, $\sigma$ chooses to increment.
This implies $N+1$ increments without resets, contradicting that $\sigma$ ensures $B(N)$.

\smallskip
Let $v_n$ given by the above property.
For every $n'$ such that $n - n(K,N) \le n' \le n$ and $p \le K$, 
for every vertex $v_{n'}$ and $u_{p,n'}$, 
there exists a memory state that leads to $u_{K+1,n - n(K,N)}$
such that from there $\sigma$ chooses to reset.
Up to renaming, we can assume that it is always the same memory state, denoted $m$.

Consider the game obtained by restricting to the first $n(K,N)$ moves from $v_n$ and excluding the topmost level;
it is equal to the game $G_{K,N}$ from $v_{n(K,N)}$ for the condition $B(N) \Until v_{n - n(K,N)}$.
Observe now that the strategy $\sigma$ restricted to this game ensures $B(N) \Until v_{n - n(K,N)}$.
Furthermore, it does not make use of the memory state $m$;
indeed, the functions have been chosen such that $(K+1) \cdot d(K+1,N) \ge n(K,N)$,
so resetting in $u_{K+1,n - n(K,N)}$ leads to the left of $v_n$.
Using the memory state $m$ at any point would allow Adam to force to reach $u_{K+1,n - n(K,N)}$
and reset from there, which would contradict the fact that $\sigma$ is normalized.
This concludes.

\section{Existence of a trade-off for thin tree arenas}
\label{sec:thin_tree}

In this section, we prove that the conjecture holds 
for the special case of thin tree arenas\footnote{The definitions of word and thin tree arenas are given in Subsection~\ref{subsec:defs_arenas}.}.

\begin{theorem}
\label{thm:thin_tree}
There exists two functions $\mem : \N^2 \to \N$ and $\alpha : \N^4 \to \N$ such that
for all $B$-parity games with $k$ counters and $d+1$ colors over thin tree arenas of width $W$ with initial vertex $v_0$,
Eve has a strategy to ensure $B(\val(v_0)^k \cdot \alpha(d,k,W,\val(v_0))) \cap \parity(\Omega)$,
with $W \cdot 3^k \cdot k! \cdot \mem(d,k)$ memory states.
\end{theorem}

The functions $\alpha$ and $\mem$ are defined as follows.
$$\alpha(d,k,W,N) = 
\begin{cases}
2N & \textrm{ if } d = 1, \\
\alpha(d-2,k+1,6W,K \cdot (N+1)^k) & \textrm{ otherwise}. \\
\end{cases}$$
$$\mem(d,k) = 
\begin{cases}
2 \cdot k! & \textrm{ if } d = 1, \\
4 \cdot \mem(d-2,k+1) & \textrm{ otherwise}. \\
\end{cases}$$

As an intermediate result, we will prove that the conjecture holds for the special case of word arenas.

\begin{theorem}
\label{thm:word}
There exists two functions $\mem : \N^2 \to \N$ and $\alpha : \N^4 \to \N$ such that
for all $B$-parity games over word arenas of width $W$ with initial vertex $v_0$,
Eve has a strategy to ensure $B(\alpha(d,W,k,\val(v_0))) \cap \parity(\Omega)$,
with $\mem(d,k)$ memory states.
\end{theorem}

\subsection{Word and thin tree arenas}
\label{subsec:defs_arenas}
A (non-labelled binary) \textit{tree} is a subset $T \subseteq \set{0,1}^*$
which is prefix-closed and non-empty.
The elements of $T$ are called \textit{nodes}, and we use the natural terminology:
for $n \in \set{0,1}^*$ and $\ell \in \set{0,1}$, the node $n \cdot \ell$ is a child of $n$,
and a descendant of $n$ if $\ell \in \set{0,1}^*$.

A (finite or infinite) branch $\pi$ is a word in $\set{0,1}^*$ or $\set{0,1}^\omega$. 
We say that $\pi$ is a branch of the tree $T$ if $\pi \subseteq T$ (or every prefix of $\pi$ belongs to $T$ when $\pi$ is infinite)
and $\pi$ is maximal satisfying this property.
A tree is called \emph{thin} if it has only countably many branches. 
For example, the full binary tree $T = \set{0,1}^*$ has uncountably many branches, therefore it is not thin.

Given a thin tree $T$, we can associate to each node $n$ a rank, denoted $\rank(n)$, 
which is a countable ordinal number, satisfying the following properties.
\begin{fact}[\cite{BojanczykIdziaszekSkrzypczak13}]\hfill
\label{fact:thin_tree}
\begin{enumerate}
	\item If $n'$ is a child of $n$, then $\rank(n') \le \rank(n)$.
	\item The set of nodes having the same rank is either a single node or an infinite branch of $T$.
\end{enumerate}	
\end{fact}

\begin{definition}
An arena is:
\begin{itemize}
	\item \emph{chronological} if there exists a function $r : V \to \N$ 
which increases by one on every edge: for all $(v,v') \in E$, $r(v') = r(v) + 1$.
	\item \emph{a word arena of width $W$} if it is chronological, and
	for all $i \in \N$, the set $\set{v \in V \mid r(v) = i}$ has cardinal at most $W$.
	\item \emph{a tree arena of width $W$} if there exists a function $R : V \to \set{0,1}^*$ such that
	\begin{enumerate}
		\item for all $n \in \set{0,1}^*$, the set $\set{v \in V \mid R(v) = n}$ has cardinal at most $W$.
		\item for all $(v,v') \in E$, we have $R(v') = R(v) \cdot \ell$ for some $\ell \in \set{0,1}$.
	\end{enumerate}
	It is a thin tree arena if $R(V)$ is a thin tree.
\end{itemize}
\end{definition}
To avoid a possible confusion: in a (thin) tree arena, ``vertices'' refers to the arena and ``nodes'' to $R(V)$,
hence if the arena has width $W$, then a node is a bundle of at most $W$ vertices.

The notions of word and tree arenas naturally appear in the study of automata over infinite words and trees.
Indeed, the acceptance games of such automata, which are used to define their semantics,
are played on word or tree arenas.
Furthermore, the width corresponds to the size of the automaton.

\subsection{Existence of a trade-off for word arenas}
We prove Theorem~\ref{thm:word} by induction on the number of colors in the parity condition.
The base case is given by B\"uchi conditions, and follows from Theorem~\ref{thm:buchi}.

Consider a $B$-parity game $\G$ with $k$ counters and $d+1$ colors over a word arena of width $W$
with initial vertex $v_0$.
Denote $N = \val(v_0)$.

We examine two cases, depending whether the least important color (\textit{i.e} the smallest) that appears is odd or even.
In both cases we construct an equivalent $B$-parity game $\G'$ using one less color;
from the induction hypothesis we obtain a winning strategy using small memory in $\G'$, 
which we use to construct a winning strategy using small memory in $\G$.

\subsubsection{Removing the least important color: the odd case}
The first case we consider is when the least important color is $1$. The technical core of the construction is motivated by the technique used in~\cite{VandenBoom11}.
Note that if this is the only color, then Eve cannot win and the result is true;
we now assume that the color $2$ also appears in the arena.
Without loss of generality we restrict ourselves to vertices reachable with $\sigma$ from $v_0$.

Consider a vertex $v$ and $T_v$ the tree of plays consistent with $\sigma$.
The strategy $\sigma$ ensures the parity condition, so in particular 
every branch in $T_v$ contains a vertex of color greater than~$1$.
We prune the tree $T_v$ by cutting paths when they first meet a vertex of color greater than~$1$.
Since the arena is finite-branching, so is $T_v$ and by Koenig's Lemma the tree obtained is finite.
Thus, to every vertex $v$, we can associate $S(v)$ a rank such that
the strategy $\sigma$ ensures that all paths from $v$ contain a vertex of color greater than $1$
before reaching the rank $S(v)$.

We define by induction an increasing sequence of integers $(S_k)_{k \in \N}$ called slices, such that $\sigma$ ensures
that between two slices, a vertex of color greater than $1$ is reached.
We first set $S_0 = S(v_0)$ (recall that $v_0$ is the initial vertex).
Assume $S_k$ has been defined, we define $S_{k+1}$ as $\max \set{S(v) \mid r(v) = S_k}$.
(Note that this is well-defined since $\set{v \mid r(v) = S_k}$ is finite.)

Now we equip $\G$ with a memory structure $\M$ of size $2$ which keeps track of the boolean information
whether or not a vertex of color greater than $1$ has been reached since the last slice.

We equip the arena $\G \times \M$ with the colouring function $\Omega'$ defined by
$$\Omega'(v,m) = 
\begin{cases}
\Omega(v) & \textrm{ if } \Omega(v) \neq 1, \\
2 & \textrm{ otherwise}.
\end{cases}$$
Remark that $\Omega'$ uses one less color than $\Omega$.

Define $L = \set{(v,1) \mid v \in S_k \textrm{ for some } k \in \N}$,
and equip $\G \times \M$ with the condition 
$B(N) \cap \parity(\Omega') \cap \safe(L)$.

\begin{lemma}\hfill
\begin{enumerate}
	\item The strategy $\sigma$ in $\G$ induces a strategy $\sigma'$ in $\G \times \M$ that 
	ensures $B(N) \cap \parity(\Omega') \cap \safe(L)$.
	\item Let $\sigma'$ be a strategy in $\G \times \M$ ensuring $B(N') \cap \parity(\Omega') \cap \safe(L)$
	with $K$ memory states,
	then there exists $\sigma$ a strategy in $\G$ that ensures $B(N') \cap \parity(\Omega)$
	with $2K$ memory states.
\end{enumerate}
\end{lemma}

\begin{proof}\hfill
\begin{enumerate}
	\item The strategy $\sigma'$ that mimics $\sigma$, ignoring the memory structure $\M$, ensures $B(N) \cap \parity(\Omega')$.
	\item Let $\sigma'$ be a strategy in $\G \times \M$ ensuring $B(N') \cap \parity(\Omega') \cap \safe(L)$
	using $\M'$ as memory structure.
	We define $\sigma$ using $\M \times \M'$ as memory structure, simply by $\sigma(v,(m,m')) = \sigma'((v,m),m')$.
	Since plays of $\sigma$ and of $\sigma'$ are in one-to-one correspondence,
	$\sigma$ ensures $B(N') \cap \parity(\Omega')$.
	Further, the safety condition satisfied by $\sigma'$ ensures that infinitely often a vertex of color greater than $1$
	is seen (specifically, between each consecutive slices), so $\sigma$ satisfies $\parity(\Omega)$.
\end{enumerate}
\end{proof}

\subsubsection{Removing the least important color: the even case}
The second case we consider is when the least important color is $0$.

We explain the intuition for the case of CoB\"uchi conditions, \textit{i.e} if there are only colors $0$ and $1$.
Let $F = \set{v \mid \Omega(v) = 0}$.
Define $X_0 = Y_0 = \emptyset$, and for $i \ge 1$:
$$\left\{
\begin{array}{l}
X_{i+1} = \WE(\safe(F) \textrm{ WeakUntil } Y_i)\\
Y_{i+1} = \WE(\reach(X_{i+1}))
\end{array}
\right.$$
The condition $\safe(F) \textrm{ WeakUntil } Y_i$ is satisfied by plays
that do not visit $F$ before $Y_i$: they may never reach $Y_i$, in which case neither $F$,
or they reach $Y_i$, in which case they did not visit $F$ before that.

We have $\bigcup_i Y_i = \WE(\co\buc(F))$.
A winning strategy based on these sets has two aims: in $X_i$ it avoids $F$ (``Safe'' mode) and in $Y_i$ it attracts to the next $X_i$ (``Attractor'' mode).
The key property is that since the arena is a word arena of width $W$ where Eve can bound the counters by $N$, 
she only needs to alternate between modes a number of times bounded by a function of $N$ and $W$. 
In other words, the sequence $(Y_i)_{i \in \N}$ stabilizes after a number of steps bounded by a function of $N$ and $W$. 
A remote variant of this bounded-alternation fact can be found in~\cite{Kupferman_Vardi}.
Hence the CoB\"uchi condition can be checked using a new counter and a B\"uchi condition, as follows.

There are two modes: ``Safe'' and ``Attractor''.
The B\"uchi condition ensures that the ``Safe'' mode is visited infinitely often.
In the ``Safe'' mode, only vertices of colors $0$ are accepted;
visiting a vertex of color $1$ leads to the ``Attractor'' mode and increments the new counter.
At any time, she can reset the mode to ``Safe''.
The counter is never reset, so to ensure that it is bounded, Eve must change modes finitely often.
Furthermore, the B\"uchi condition ensures that the final mode is ``Safe'',
implying that the CoB\"uchi condition is satisfied.

For the more general case of parity conditions, the same idea is used,
but as soon as a vertex of color greater than $1$ is visited,
then the counter is reset.

Define $\G'$:
$$V' = 
\begin{cases}
\VE' = \VE \times \set{A,S} \ \cup\ \overline{V} \\
\VA' = \VA \times \set{A,S}\ .
\end{cases}$$ 
After each edge followed, Eve is left the choice to set the flag to $S$.
The set of choice vertices is denoted $\overline{V}$.
We define $E'$ and the counter actions.
$$E' = 
\begin{cases}
(v,A) \xrightarrow{\ c(v,v'),\varepsilon\ } \overline{v'} & \textrm{if } (v,v') \in E, \\
(v,S) \xrightarrow{\ c(v,v'),\varepsilon\ } (v',S) & \textrm{if } (v,v') \in E \textrm{ and } \Omega(v') = 0,\\
(v,S) \xrightarrow{\ c(v,v'),i\ } (v',A) & \textrm{if } (v,v') \in E \textrm{ and } \Omega(v') = 1,\\
(v,S) \xrightarrow{\ c(v,v'),r\ } (v',S) & \textrm{if } (v,v') \in E \textrm{ and } \Omega(v') > 1,\\
\overline{v} \xrightarrow{\ \varepsilon\ } (v,A) \textrm{ and } \overline{v} \xrightarrow{\ \varepsilon\ } (v,S)
\end{cases}$$

Equip the arena $\G'$ with the colouring function $\Omega'$ defined by
$$\Omega'(v,m) = 
\begin{cases}
1 & \textrm{ if } m = A, \\
2 & \textrm{ if } \Omega(v) = 0 \textrm{ and } m = S, \\
\Omega(v) & \textrm{ otherwise}.
\end{cases}$$
We do not color the choice vertices, which does not matter as all plays contain infinitely many non-choice vertices; 
we could give them the least important color, that is $1$.
Remark that $\Omega'$ uses one less color than $\Omega$, since no vertices have color $0$ for $\Omega'$.

Before stating and proving the equivalence between $\G$ and $\G'$,
we formalise the property mentioned above, that in word arenas Eve does not need to alternate
an unbounded number of times between the modes ``Safe'' and ``Attractor''.

\begin{lemma}
\label{lem:word_collapse}
Let $G$ be a word arena of width $W$, and a subset $F$ of vertices such that
every path in $G$ contains finitely many vertices in $F$.
Define the following sequence of subsets of vertices
$X_0 = \emptyset$, and for $i \ge 0$
$$\left\{
\begin{array}{l}
X_{2i+1} = \left\{v \left| 
\begin{array}{c}
\textrm{all paths from } v \textrm{ contain no vertices in } F \\
\textrm{ before the first vertex in } X_{2i}, \textrm{ if any}
\end{array}\right.\right\},\\[1.5em]
X_{2i+2} = \left\{v \left| 
\begin{array}{c}
\textrm{all paths from } v \textrm{ are finite or lead to } X_{2i+1}
\end{array}\right.\right\}.
\end{array}
\right.$$
We have $X_0 \subseteq X_1 \subseteq X_2 \cdots$,
and $X_{2W}$ covers the whole arena.
\end{lemma}

\begin{proof}
We first argue that the following property, denoted $(\dagger)$, holds:
``for all $i \ge 0$, if $X_{2i}$ does not cover the whole arena,
then $X_{2i+1} \setminus X_{2i-1}$ contains an infinite path''.
(For technical convenience $X_{-1} = \emptyset$.)

Let $v \notin X_{2i}$.
We consider $G_v$ where $v$ is the initial vertex, and prune it by removing the vertices from $X_{2i-1}$, 
as well as vertices which do not have an infinite path after removing $X_{2i-1}$;
denote by $G'_v$ the graph obtained.
Note that for any $u \notin X_{2i}$, the vertex $u$ belongs to $G'_v$,
so $G'_v$ contains an infinite path.
We claim that there exists a vertex $v'$ in $G'_v$ such that all paths
from $v'$ contain no vertices $F$.
Indeed, assume towards contradiction that from every node in $G'_v$,
there exists a path to a vertex in $F$.
Then there exists a path that visits infinitely many vertices in $F$,
contradicting the assumption on $G$.
Any infinite path from $v'$ is included into $X_{2i+1} \setminus X_{2i-1}$,
hence the latter contains an infinite path.

\medskip
We conclude using $(\dagger)$: assume towards contradiction that $X_{2W}$ does not cover the whole arena.
Then $G$ contains $W+1$ pairwise disjoint paths, contradicting that it has width $W$.
\hfill\qed\end{proof}

\begin{lemma}\hfill
\begin{enumerate}
	\item There exists a strategy $\sigma'$ in $\G'$ that 
ensures $B(W \cdot (N+1)^k) \cap \parity(\Omega')$.
	\item Let $\sigma'$ be a strategy in $\G'$ ensuring $B(N') \cap \parity(\Omega')$
with $K$ memory states,
then there exists $\sigma$ a strategy in $\G$ that ensures $B(N') \cap \parity(\Omega)$
with $2K$ memory states.
\end{enumerate}
\end{lemma}

\begin{proof}\hfill
\begin{enumerate}
	\item Thanks to Lemma~\ref{lem:finite_memory_trivial}, there exists a strategy $\sigma$ 
in $\G$ ensuring $B(N) \cap \parity(\Omega)$ using a memory structure $\M$ of size $(N+1)^k$.
We construct a strategy $\sigma'$ in $\G'$ by mimicking $\sigma$.
We now explain when does $\sigma'$ chooses to set the flag to value $S$, \textit{i.e} sets the ``Safe'' mode.

We consider the arena $\G \times \M$, it is a word arena of width $W \cdot (N+1)^k$,
and restrict it to the moves prescribed by $\sigma$, obtaining the word arena $\G_\sigma$ of width $W \cdot (N+1)^k$.
Without loss of generality we restrict $\G_\sigma$ to vertices reachable with $\sigma$ from the initial vertex $(v_0,0)$.
Consider a vertex $v$ of color $0$ or $1$, and $G_\sigma^v$ the word arena obtained by considering $v$ as initial vertex
and pruned by cutting paths when they first meet a vertex of color greater than~$1$.
Since the strategy $\sigma$ ensures that the parity condition is satisfied,
every infinite path in $G_\sigma^v$ contains finitely many vertices of color~$1$.
Relying on Lemma~\ref{lem:word_collapse} for the word arena $G_\sigma^v$ and $F$ the set of vertices of color~$1$, 
we associate to each vertex $v'$ in $G_\sigma^v$ a rank, which is a number between $1$ and $2W \cdot M$,
the minimal $i$ such that $v' \in X_i(G_\sigma^v)$.

Now consider a play consistent with $\sigma$, and a suffix of this play starting in a vertex $v$ of color $0$ or $1$.
By definition, from this position on, the rank (with respect to $v$) is non-increasing until a vertex of color greater than~$1$ is visited,
if any. Furthermore, if the rank is even then no vertices of color~$1$ are visited,
and the rank does not remain forever odd.

The strategy $\sigma'$ in $\G'$ mimics $\sigma$, and at any point of a play
remembers the first vertex $v$ that has not been followed by a vertex of color greater than~$1$.
As observed above, the rank with respect to $v$ is non-increasing; the strategy $\sigma'$ switches
to the ``Safe'' mode when the rank goes from even to odd.
By definition, the new counter is incremented only when the rank goes from odd to even,
which happens at most $W \cdot (N+1)^k$ times,
and it is reset when a vertex of color greater than $1$ is visited,
so $\sigma'$ ensure that it remains bounded by~$W \cdot (N+1)^k$.

Also, since $\sigma$ ensures to bound the counters by $N$, then so does $\sigma'$.
For the parity condition, there are two cases.
Consider a play consistent with~$\sigma'$.
Either from some point onwards the only colors seen are $0$ and $1$ (with respect to $\Omega$),
then the new counter is not reset after this point, but it is incremented 
only when the rank decreases from odd to even, which corresponds to switches of mode from ``Safe'' to ``Attractor''.
Since this counter is bounded, the mode stabilizes, 
which by definition of the ranks imply that the stabilized rank is odd, so the mode is ``Safe'',
and from there on only vertices of color $0$ (with respect to $\Omega$)
are visited, hence $\parity(\Omega')$ is satisfied.
Or infinitely many vertices of color greater than $1$ are seen (with respect to $\Omega$), but since
they coincide for $\Omega$ and $\Omega'$, the condition $\parity(\Omega')$ is satisfied.

It follows that $\sigma'$ ensures $B(W \cdot (N+1)^k) \cap \parity(\Omega')$.

	\item Let $\sigma'$ be a strategy in $\G'$ ensuring $B(N') \cap \parity(\Omega')$
using $\M'$ as memory structure of size $K$.
We construct $\sigma$ that mimics $\sigma'$; 
to this end, we need a memory structure which simulates both $\M'$ and the boolean flag,
of size $2K$.
By definition, plays of $\sigma$ and plays of $\sigma'$ are in one-to-one correspondence,
so $\sigma$ ensures $B(N')$.
For the parity condition, there are two cases.
Consider a play consistent with $\sigma'$.
Either from some point onwards the only colors seen are $0$ and $1$ (with respect to $\Omega$),
then the new counter is not reset after this point, but it is incremented each time
the mode switches from ``Safe'' to ``Attractor'';
since this counter is bounded, the mode stabilizes, and since the play in $\G'$
satisfies $\parity(\Omega')$, the stabilized mode is ``Safe'',
implying that from there on only vertices of color $0$ (with respect to $\Omega$)
are visited, hence satisfy $\parity(\Omega)$.
Or infinitely many vertices of color greater than $1$ are seen (with respect to $\Omega$), but since
they coincide for $\Omega$ and $\Omega'$, the condition $\parity(\Omega)$ is satisfied.
\end{enumerate}
\hfill\qed\end{proof}

\subsection{Extending to thin tree arenas}
In this subsection, we extend the results for word arenas to thin tree arenas, proving Theorem~\ref{thm:thin_tree}.

Consider a $B$-parity game $\G$ with $k$ counters and $d+1$ colors over a thin tree arena of width $W$ with initial vertex $v_0$.
Define $N = \val(v_0)$.
Let $R : V \to \set{0,1}^*$ witnessing that $\G$ is a thin tree arena.
We rely on the decomposition of the thin tree $R(V)$ to locally replace
$\sigma$ by strategies using small memory given by Theorem~\ref{thm:word}.

It follows from Fact~\ref{fact:thin_tree} that along a play, the rank is non-increasing and decreases only finitely many times.
Since the parity condition is prefix-independent, if for each rank Eve plays a strategy ensuring the parity condition,
then the resulting strategy ensures the parity condition; however, a closer attention to the counters is required.

We summarize counter actions as follows: let $w \in (\set{\varepsilon,i,r}^k)^*$,
its summary $\summ(w) \in \set{\varepsilon,i,r}^k$ is, for each counter,
$r$ if the counter is reset in $w$, $i$ if the counter is incremented by not reset in $w$, and $\varepsilon$ otherwise.
\begin{fact}
\label{fact:summary}
Consider $w = w_1 w_2 \cdots w_n w_\infty$, where $w_1,\ldots,w_n \in (\set{\varepsilon,i,r}^k)^*$ 
and $w_\infty \in (\set{\varepsilon,i,r}^k)^\omega$.
Denote $u = \summ(w_1) \summ(w_2) \cdots \summ(w_n) \summ(w_\infty)$, then:
\begin{enumerate}
	\item $\val(u) \le \val(w)$,
	\item if for all $i \in \set{1,\ldots,n,\infty}$ we have $\val(w_i) \le N'$ and $\val(u) \le N$,
	then $\val(w) \le N \cdot N'$.
\end{enumerate} 
\end{fact}

We define a $B$-game $\G'$, where the plays that remain in vertices of the same rank are summarized in one step.
It has $k$ counters (as does $\G$).
Let $\rank(V)$ denote the set of ranks (subset of the countable ordinals), 
and $\S$ the set of all strategies in $\G$ ensuring $B(N) \cap \parity(\Omega)$.
Define:
$$V = \left\{
\begin{array}{l}
\VE = \rank(V) \times \set{1,\ldots,W} \times \set{\varepsilon,i,r}^k \\
\VA = \rank(V) \times \set{1,\ldots,W} \times \S
\end{array}\right.$$
We explain how a couple $(\nu,\ell) \in \rank(V) \times \set{1,\ldots,W}$ uniquely determines a vertex in $\G$.
First, the rank $\nu$ corresponds in $R$ either to a node or to an infinite branch,
in the second case we consider the first node in this branch.
Second, the component $\ell$ identifies a vertex in this node.

We say that $(\nu,\ell',a)$ is an outcome of $(\mu,\ell,\sigma)$ if there exists a play from the vertex corresponding to $(\mu,\ell)$ consistent with $\sigma$
ending in the vertex corresponding to $(\nu,\ell')$ whose summarized counter actions are $a$.
$$E = 
\begin{cases}
\set{((\nu,\ell,a), (\nu,\ell,\sigma)) \mid \ell,\nu,\sigma} \hfill \textrm{counter action: } a \\
\set{((\mu,\ell,\sigma), (\nu,\ell',a)) \mid \textrm{ if } (\nu,\ell',a) \textrm{ is an outcome of } (\mu,\ell,\sigma)}
\end{cases}$$
By definition, $\G'$ is well-founded \textit{i.e} there are no infinite plays in $\G'$.

We first argue that there exists a strategy in $\G'$ ensuring $B(N)$.
Indeed, it is induced by the strategy $\sigma$.
A play consistent with this strategy is of the form 
$u = \summ(w_1) \summ(w_2) \cdots \summ(w_n) \summ(w_\infty)$,
where $w = w_1 w_2 \cdots w_n w_\infty$ is a play consistent with $\sigma$,
following the notations of Fact~\ref{fact:summary}.
This fact, item 1., implies that $\val(u) \le \val(w)$, so $\val(u) \le N$.
Hence the induced strategy ensures $B(N)$.

\begin{lemma}
\label{lem:well_founded}
For all $B$-games with $k$ counters over a well-founded arena with initial vertex $v_0$,
Eve has a strategy to ensure $B(\val(v_0)^k)$ with $k!$ memory states.
\end{lemma}

Thanks to Lemma~\ref{lem:well_founded}, there exists a strategy $\sigma'$ in $\G'$ ensuring $B(N^k)$ and using the memory structure $\M = (M,m_0,\up)$, of size $k!$.
We construct a strategy in $\G$ ensuring $B(N^k \cdot \alpha(d,W,k,N)) \cap \parity(\Omega)$ using $W \cdot 3^k \cdot k! \cdot \mem(d,k)$ memory states.
The memory structure is the product of four memory structures: 
\begin{itemize}
	\item a memory structure that keeps track of the $\ell \in \set{1,\ldots,W}$ we used when entering the current rank,
	\item a memory structure that keeps track of the summary since we entered the current rank, of size $3^k$,
	\item the memory structure $\M$,
	\item a memory structure of size $\mem(d,k)$, which is used to simulate the strategies obtained from Theorem~\ref{thm:word}.
\end{itemize}
Consider $\nu \in \rank(V)$ that corresponds to an infinite branch of $R(V)$.
For every $\ell \in \set{1,\ldots,W}$ and $m \in M$, the strategy $\sigma'$ picks a strategy $\sigma'(\nu,\ell,m)$
to play in this infinite branch, ensuring $B(N) \cap \parity(\Omega)$.
When playing this strategy, two scenarios are possible: either the play stays forever in the infinite branch,
or an outcome is selected and the game continues from there.
Consider the game obtained by starting from the vertex corresponding to $(\nu,\ell)$ and restricted to the infinite branch
corresponding to $\nu$, plus a vertex for each outcome.
It is played over a word arena of width $W$ with $d+1$ colors and $k$ counters.
Denote by $O(\sigma'(\nu,\ell,m))$ the set of outcomes of this game that are not consistent with $\sigma'(\nu,\ell,m)$. 
The strategy $\sigma'(\nu,\ell,m)$ ensures $B(N) \cap \parity(\Omega) \cap \safe(O(\sigma'(\nu,\ell,m)))$.
Thanks to Theorem~\ref{thm:word}, there exists a strategy $\sigma(\nu,\ell,m)$ ensuring $B(\alpha(d,W,k,N)) \cap \parity(\Omega) \cap \safe(O(\sigma'(\nu,\ell,m)))$
using $\mem(d,k)$ memory states.

The strategy $\sigma$ simulates the strategies $\sigma(\nu,\ell,m)$ in the corresponding parts of the game.
Observe that this requires to keep track of both the value $\ell$ and the summary of the current rank,
which is done by the memory structure.
We argue that $\sigma$ ensures $B(N^k \cdot \alpha(d,W,k,N)) \cap \parity(\Omega)$.
A play consistent with this strategy is of the form 
$w = w_1 w_2 \cdots w_n w_\infty$,
where for all $i \in \set{1,\ldots,n,\infty}$, we have $\val(w_i) \le \alpha(d,W,k,N)$.
Denote $u = \summ(w_1) \summ(w_2) \cdots \summ(w_n) \summ(w_\infty)$,
we have $val(u) \le N^k$,
since it corresponds to a play consistent with $\sigma'$.
It follows from Fact~\ref{fact:summary}, item 2., that $\val(w) \le N^k \cdot \alpha(d,W,k,N)$.
Hence the strategy $\sigma$ ensures $B(N^k \cdot \alpha(d,W,k,N)) \cap \parity(\Omega)$.

\section*{Conclusion}
We studied the existence of a trade-off between bounds and memory in games with counters,
as conjectured by Colcombet and Loeding.
We proved that there is no such trade-off in general, but that under some structural restrictions,
as thin tree arenas, the conjecture holds.

We believe that the conjecture holds for all tree arenas, which would imply the decidability of cost MSO
over infinite trees. 
A proof of this result would probably involve advanced combinatorial arguments, and require a deep understanding
of the structure of tree arenas.

\section*{Acknowledgments}
The unbounded number of fruitful discussions we had with Thomas Colcombet and Miko{\l}aj Boja{\'n}czyk
made this paper possible.

\bibliographystyle{plain}
\bibliography{bib}

\begin{thebibliography}{10}

\bibitem{BCKPV14}
Achim Blumensath, Thomas Colcombet, Denis Kuperberg, Pawe{\l} Parys, and
  Michael {Vanden Boom}.
\newblock Two-way cost automata and cost logics over infinite trees.
\newblock In {\em CSL-LICS}, pages 16--26, 2014.

\bibitem{BlumensathOttoWeyer14}
Achim Blumensath, Martin Otto, and Mark Weyer.
\newblock Decidability results for the boundedness problem.
\newblock {\em Logical Methods in Computer Science}, 10(3), 2014.

\bibitem{Bojanczyk04}
Miko{\l}aj Boja{\'n}czyk.
\newblock A bounding quantifier.
\newblock In {\em CSL}, pages 41--55, 2004.

\bibitem{BojanczykColcombet06}
Miko{\l}aj Boja{\'n}czyk and Thomas Colcombet.
\newblock Bounds in {$\omega$}-regularity.
\newblock In {\em LICS}, pages 285--296, 2006.

\bibitem{BojanczykIdziaszekSkrzypczak13}
Miko{\l}aj Boja{\'n}czyk, Tomasz Idziaszek, and Micha{\l} Skrzypczak.
\newblock Regular languages of thin trees.
\newblock In {\em STACS}, pages 562--573, 2013.

\bibitem{Colcombet09}
Thomas Colcombet.
\newblock The theory of stabilisation monoids and regular cost functions.
\newblock In {\em ICALP (2)}, pages 139--150, 2009.

\bibitem{Colcombet13}
Thomas Colcombet.
\newblock Fonctions r{\'e}guli{\`e}res de co{\^u}t.
\newblock Habilitation Thesis, 2013.

\bibitem{Colcombet13a}
Thomas Colcombet.
\newblock Regular cost functions, part {I:} logic and algebra over words.
\newblock {\em Logical Methods in Computer Science}, 9(3), 2013.

\bibitem{ColcombetLoeding08}
Thomas Colcombet and Christof L{\"o}ding.
\newblock The non-deterministic {M}ostowski hierarchy and distance-parity
  automata.
\newblock In {\em ICALP (2)}, pages 398--409, 2008.

\bibitem{ColcombetLoeding10}
Thomas Colcombet and Christof L{\"o}ding.
\newblock Regular cost functions over finite trees.
\newblock In {\em LICS}, pages 70--79, 2010.

\bibitem{LNCS2500}
Erich Gr{\"a}del, Wolfgang Thomas, and Thomas Wilke, editors.
\newblock {\em Automata, Logics, and Infinite Games}, volume 2500 of {\em
  Lecture Notes in Computer Science}. Springer, 2002.

\bibitem{Hashiguchi90}
Kosaburo Hashiguchi.
\newblock Improved limitedness theorems on finite automata with distance
  functions.
\newblock {\em Theoretical Computer Science}, 72(1):27--38, 1990.

\bibitem{Kirsten05}
Daniel Kirsten.
\newblock Distance desert automata and the star height problem.
\newblock {\em ITA}, 39(3):455--509, 2005.

\bibitem{Kupferman_Vardi}
Orna Kupferman and Moshe~Y. Vardi.
\newblock Weak alternating automata are not that weak.
\newblock In {\em 5th Israeli Symposium on Theory of Computing and Systems},
  pages 147--158. {IEEE} Computer Society Press, 1997.

\bibitem{Leung91}
Hing Leung.
\newblock Limitedness theorem on finite automata with distance functions: An
  algebraic proof.
\newblock {\em Theoretical Computuer Science}, 81(1):137--145, 1991.

\bibitem{Rabin69}
Michael~O. Rabin.
\newblock Decidability of second-order theories and automata on infinite trees.
\newblock {\em Transactions of the AMS}, 141:1--23, 1969.

\bibitem{Simon94}
Imre Simon.
\newblock On semigroups of matrices over the tropical semiring.
\newblock {\em ITA}, 28(3-4):277--294, 1994.

\bibitem{VandenBoom11}
Michael {Vanden Boom}.
\newblock Weak cost monadic logic over infinite trees.
\newblock In {\em MFCS}, pages 580--591, 2011.

\end{thebibliography}

\end{document}